\documentclass{lmcs}
\pdfoutput=1

\usepackage{lastpage}
\lmcsdoi{15}{1}{21}
\lmcsheading{}{\pageref{LastPage}}{}{}%
{Jul.~04,~2018}{Mar.~05,~2019}{}

\usepackage{amsmath}
\usepackage{amssymb}
\usepackage{amsthm}
\usepackage{cancel}
\usepackage{microtype}
\usepackage{color}
\usepackage{hyperref}
\usepackage{subcaption}
\usepackage{graphicx}
\usepackage{subcaption}
\usepackage{tikz}
\usepackage{pgfplots}
\pgfplotsset{compat=1.11}

\newcommand{\NN}{\mathbb{N}}
\newcommand{\RR}{\mathbb{R}}
\newcommand{\blockset}{B}
\newcommand{\block}{b}
\newcommand{\traceset}{T}
\newcommand{\trace}{t}
\newcommand{\len}{l}

\newcommand{\algorithm}{P}
\newcommand{\algorithmprime}{Q}
\newcommand{\controlstateset}{S}
\newcommand{\controlstate}{s}
\newcommand{\initialstate}{i}
\newcommand{\capacity}{n}
\newcommand{\upcontrol}{tr}
\newcommand{\evict}{evict}
\newcommand{\miss}{miss}
\newcommand{\contentset}{C}
\newcommand{\content}{c}
\newcommand{\invalid}{\bot}
\newcommand{\configurations}{G}
\newcommand{\config}{g}
\newcommand{\initialconf}{\initialstate g}
\newcommand{\upconfig}{update}

\newcommand{\leakrat}{r}

\newcommand{\sizeof}[1]{\left| #1 \right|}


\begin{document}

\title{On the Incomparability of Cache Algorithms in Terms of Timing Leakage}

\author{Pablo Ca{\~n}ones}
\address{IMDEA Software Institute and Universidad Polit{\'e}cnica de Madrid, Madrid, Spain}
\email{pablo.canones@imdea.org}

\author{Boris K{\"o}pf}
\address{IMDEA Software Institute, Madrid, Spain and Microsoft Research, Cambridge, UK}
\email{boris.koepf@microsoft.com}

\author{Jan Reineke}
\address{Saarland University, Saarbr{\"u}cken, Germany}
\email{reineke@cs.uni-saarland.de}

\begin{abstract}
Modern computer architectures rely on caches to reduce the latency gap
between the CPU and main memory. While indispensable for performance, caches pose a serious threat
to security because they leak information about memory access patterns
of programs via execution time. 

In this paper, we present a novel approach for reasoning about the
security of cache algorithms with respect to timing leaks. The basis
of our approach is the notion of {\em leak competitiveness}, which
compares the leakage of two cache algorithms on every possible
program. Based on this notion, we prove the following two results:

First, we show that leak competitiveness is \emph{symmetric} in the
cache algorithms.  This implies that no cache algorithm dominates
another in terms of leakage via a program's total execution time. This is in contrast
to performance, where it is known that such dominance
relationships exist.

Second, when restricted to caches with finite control, the
leak-competitiveness relationship between two cache algorithms is
either asymptotically linear or constant. No other shapes are
possible.
\end{abstract}

\maketitle

\section{Introduction}
Modern computer architectures rely on caches to reduce the latency gap between the CPU and main memory. Accessing data that is cached (a cache hit) can be hundreds of CPU cycles faster than accessing data that needs to be fetched from main memory (a cache miss), which translates into significant performance gains. 

While caches are indispensable for performance, they pose a serious threat to security. An attacker who can distinguish between cache hits and misses via timing measurements can learn information about the memory access pattern of a victim's program. This side channel has given rise to a large number of documented attacks, e.g.~\cite{aciiccmez2006trace,aciiccmez2007cache,Bernstein05cache-timingattacks,GullaschBK11,kocher2018spectre,lipp2018meltdown,
LiuYGHL15,osvikshamir06cache,YaromF14}.

From a security point of view it would be ideal to completely eliminate cache side channels by design, as in~\cite{TiwariOLVLHKCS11,zwsm15}. Unfortunately, such conservative approaches also partially void the performance benefits of caches. In practice, one usually seeks to identify a trade-off between security and performance, which requires comparing different cache designs in terms of their security and performance properties.

While performance analysis of cache designs is an established field~\cite{al2004performance,dorrigiv2010alternative} there are only few approaches concerned with analyzing their security. Examples are~\cite{he2017secure}, which analyzes the effect of the size of the cache and how it is shared between different agents, and \cite{canones2017}, which measures the security of cache algorithms with respect to adversaries who can gather information about a victim's computation by probing the state of a shared cache.

In this paper, we present a novel approach for evaluating the security of caches. More precisely, we focus on the amount of information that a cache algorithm leaks to an adversary that can measure a program's overall execution time. Leakage through a program's overall execution time is practically relevant because it can be exploited remotely, and conceptually interesting because the notions of security and performance are tightly coupled -- even though there are more powerful ways to spy on a program via shared caches~\cite{YaromF14}.

The basis of our approach is a novel notion for comparing the leakage of cache algorithms, which we call {\em leak competitiveness}. Leak competitiveness is inspired by {\em competitiveness}, which is a standard notion for comparing the performance of online algorithms, and in particular cache algorithms~\cite{reineke2008relative,sleator1985amortized}. However, whereas competitive performance analysis compares cache algorithms on individual traces, leak competitiveness compares the leakage of cache algorithms on {\em sets} of traces, which accounts for the fact that information flow is a hyperproperty~\cite{clarkson2010hyperproperties}.

The central contribution of this paper is a characterization of the possible leak-competitiveness relationships between any two cache algorithms:%
\begin{itemize}
\item We find that leak competitiveness is \emph{symmetric} in the cache algorithms. 
	This implies that no cache algorithm dominates another in terms of leakage via execution time.
	Note that this is in contrast to performance, where it is known that such dominance relationships exist~\cite{reineke2008relative}.
	This result holds for a very general class of deterministic cache algorithms, including fully-associative caches, set-associative caches with arbitrary replacement policies, and even rather exotic caches such as skewed-associative caches.
\item If we restrict our attention to caches with finite control, which is natural for hardware-based cache implementations, the leak-competitiveness relationship between two cache algorithms is either asymptotically linear in the length of the program execution or it is constant. No other shapes are possible. 
\end{itemize}

The proofs of these results are based on three intermediate steps that are of independent interest. 
\begin{enumerate}
\item The first is to show that a pair of traces of memory accesses precisely characterizes the leak competitiveness relationship between any two cache algorithms.
\item The second step is to show that we can actually identify a {\em single} trace of memory accesses for which the difference in number of misses between both algorithms matches their leak competitiveness to within a factor of 2. 
This is surprising in the light that leakage is a hyperproperty, i.e., it requires {\em sets} of traces to express. 

\item The third step is to define a congruence on the cache contents of algorithms with finite control -- but potentially infinite data -- and to show that the resulting quotient is {\em finite}.
Our characterization of leak competitiveness follows from the observation that, if the trace that witnesses the leak competitiveness is large enough, it will visit multiple congruent cache states, i.e. contain a cycle in the quotient. We then use a pumping argument to obtain a linear lower bound on the leak competitiveness from this cycle.
\end{enumerate}

\paragraph{\bf Organization of the paper} The remainder of this paper is structured as follows. In Section~\ref{sec:prelim} we introduce a general model of deterministic cache algorithms.  In Section~\ref{sec:leakrat} we introduce leak competitiveness and leak ratio, based on which we present our main results in Section~\ref{sec:charleakrat}.  Sections~\ref{sec:sequences}--\ref{sec:lowbound} present the proof of our results, following the structure outlined above. We present related work in Section~\ref{sec:related} before we conclude in Section~\ref{sec:conclusion}.

\section{Preliminaries}\label{sec:prelim}
Caches are fast but small memories that store a subset of the main memory's contents to bridge the latency gap between the CPU and the main memory.  To profit from spatial locality and to reduce management overhead, main memory is logically partitioned into a set $\blockset$ of memory blocks. Each block is cached as a whole in a cache line of the same size.  When accessing a memory block, the cache logic has to determine whether the block is stored in the cache (``cache hit'') or not (``cache miss''). In the case of a miss, the cache algorithm decides which memory block to evict and replace by a new one.

\begin{defi}
A {\em cache algorithm} (or \emph{algorithm}) is a tuple
$$\algorithm = (\controlstateset_\algorithm, \initialstate_\algorithm, \capacity_\algorithm, \upcontrol_\algorithm, \evict_\algorithm),$$ which consists of the following components:
\begin{itemize}
	\item The set of \emph{control states}, $\controlstateset_\algorithm$.
	\item The \emph{initial} control state, $\initialstate_\algorithm \in \controlstateset_\algorithm$.
	\item The \emph{capacity} of the cache, $\capacity_\algorithm \in \NN$.
	\item The \emph{transition} function, $\upcontrol_\algorithm : \controlstateset_\algorithm \times \{0, \dots, \capacity_\algorithm-1\} \rightarrow \controlstateset_\algorithm$, that, upon a hit to one of its $\capacity_\algorithm$ cache lines, determines the new control state of the cache.
	\item The \emph{evict} function, $\evict_\algorithm : \controlstateset_\algorithm \times \blockset \rightarrow \controlstateset_\algorithm \times \{0, \dots, \capacity_\algorithm-1\}$, that, upon a miss, determines the new control state of the cache and the cache line to evict. 
\end{itemize}
\end{defi}
During runtime a \emph{cache configuration} consists of the cache's control state and of its current \emph{content}.
The content is captured by a function $\content : \contentset_\algorithm = \{0, \dots, \capacity_\algorithm-1\} \rightarrow \blockset \cup \{\invalid\}$ that maps each cache line to the memory block it holds, or $\invalid$ if the line is invalid.
A cache configuration $\config=(\controlstate, \content) \in \configurations_\algorithm = \controlstateset_\algorithm \times \contentset_\algorithm$ is updated as follows upon a memory access:
\begin{equation}
\label{eq:update}
\upconfig_\algorithm((\controlstate,\content), \block) := 	\begin{cases}
							(\controlstate', \content) & \textit{if } \exists j: \content(j) = \block \wedge \controlstate' = \upcontrol_\algorithm(\controlstate,j),\\
							(\controlstate', \content[j \leftarrow \block]) & \textit{if } \forall k: \content(k) \neq \block \wedge (\controlstate', j) = \evict_\algorithm(\controlstate, \block).
						\end{cases}
\end{equation}
Upon a hit, the update function is used to obtain the new control state.
Upon a miss, the accessed block replaces one of the cached blocks, determined by the evict function.

The above definition of a cache algorithm is quite general: 
it captures arbitrary deterministic caches that operate on a bounded capacity buffer.
This includes direct-mapped, set-associative, fully-associative caches, and even skewed-associative caches with arbitrary deterministic replacement policies.
Well-known \emph{deterministic} replacement policies which fit our model are \emph{least-recently used} (LRU), used in various Freescale processors such as the MPC603E and the TriCore17xx, as well as the recent Kalray MPPA~256; \emph{pseudo-LRU} (PLRU), a cost-efficient variant of LRU, used in the Freescale MPC750 family and multiple Intel microarchitectures; \emph{most-recently used} (MRU), also known as \emph{not most-recently used} (NMRU), another cost-efficient variant of LRU, used in the Intel Nehalem; \emph{first-in first-out} (FIFO), also known as \textsc{Round Robin}, used in several ARM and Freescale processors such as the ARM922 and the Freescale MPC55xx family; Pseudo-Round Robin, used in the NXP Coldfire 5307.

\emph{Notation:}
The update function is lifted to traces $\trace \in \blockset^*$ of blocks recursively as follows:
\begin{align*}
\upconfig_\algorithm((\controlstate,\content), \epsilon) & := (\controlstate,\content),\\
\upconfig_\algorithm((\controlstate,\content), \block\circ \trace) & := \upconfig_\algorithm(\upconfig_\algorithm((\controlstate,\content),\block), \trace).
\end{align*}
						
The number of misses $\algorithm((\controlstate,\content),\trace)$ of an algorithm $\algorithm$ on a trace $\trace \in \blockset^*$ starting in configuration $(\controlstate,\content)$ is determined recursively as follows:
\begin{align*}
\algorithm((\controlstate,\content), \epsilon) & := 0,\\
\algorithm((\controlstate,\content), \block\circ \trace) & := \miss(\block,\content) + \algorithm(\upconfig_\algorithm(\block,(\controlstate,\content)), \trace),
\end{align*}
where
$\miss(\block,\content) = (\forall j: \content(j) \neq \block~?~1 : 0)$.
				
We use $\algorithm(\trace)$ as a shortcut for $\algorithm((\initialstate_\algorithm,\lambda j.\invalid), \trace)$, i.e., the number of misses on the trace $\trace$ when starting in the initial configuration of the cache.
Also, $\algorithm(\trace, \trace')$ is a shortcut for $\algorithm(\trace\trace')-\algorithm(\trace)$, i.e., the number of misses on the suffix~$\trace'$.  

\section{Leak Ratio}\label{sec:leakrat}
In this section we define a measure for comparing the security of cache algorithms, which we call the leak ratio.  The leak ratio is inspired
by quantitative notions of security in information flow analysis~\cite{koepfbasin07,smith09} and by the notion of relative miss competitiveness~\cite{reineke2008relative} from the real-time systems community. We revisit both notions first.

\subsection{Relative Miss Competitiveness}
Relative competitiveness~\cite{reineke2008relative} is a notion for comparing the worst-case performance of two cache algorithms. It is based on the classic notion of competitiveness~\cite{sleator1985amortized}, which compares an online algorithm with the optimal offline algorithm. 
Below, we reproduce a slightly simplified version of the definition of relative competitiveness from~\cite{reineke2008relative}:

\begin{defi}\label{def:miss-competitive}
  For $\leakrat\in\RR_{>0}$, we say that a algorithm $\algorithm$ is {\em $\leakrat$-miss-competitive relative to} algorithm $\algorithmprime$ if there exists $c\in\RR_{>0}$ such that
\begin{equation*}
\algorithm(\trace)\leq \leakrat\cdot \algorithmprime(\trace)+c,
\end{equation*}
for all traces $\trace \in\blockset^*$.
\end{defi}
\begin{exa}\label{ex:miss-competitive}
LRU of associativity 4 is 1-miss-competitive relative to FIFO of associativity 2. 
On the other hand, FIFO of associativity 2 is not $\leakrat$-miss-competitive to LRU of associativity 4 for any $\leakrat$.
Therefore, LRU of associativity 4 outperforms FIFO of associativity 2 in number of misses.
See~\cite{reineke2008relative} for details and more examples.
\end{exa}

\subsection{Leak Competitiveness}\label{sec:leak_ratio}

We next introduce a notion based on relative competitiveness that compares the amount of information that two cache algorithms leak via their timing behavior. We begin by recalling basic concepts from quantitative information-flow analysis.

\subsubsection{Quantifying Leaks}

As is common in side-channel analysis based on quantitative information-flow~\cite{koepfbasin07}, we quantify the amount of information a system leaks in terms of the number of observations an adversary can make. This number represents an upper bound on entropy loss, for different notions of entropy including Shannon entropy, min-entropy~\cite{smith09}, or g-vulnerability~\cite{m2012measuring}. Each of those notions of entropy is associated with an interpretation in terms of security. For example, using min-entropy as a basis for the interpretation, a bound on the number of observations corresponds to an upper bound on the factor by which guessing becomes easier through side-channel information. 

For formalizing a program's leakage through cache timing effects, we abstract the program in terms of the set $\traceset$ of traces of memory accesses it can perform. We always consider traces of finite length $\len$, hence $\traceset\subseteq\blockset^\len$.  We capture the attacker's observation of a program execution as the number of cache misses produced by the corresponding trace. If $\len$ is known to the adversary, then she can deduce the number of misses from the program's overall execution time, due to the large latency gap between cache hits and cache misses\footnote{We consider ``noiseless'' attacks where the attacker is able to obtain the maximum amount of information from the cache. This assumption is common in the literature on security to safely over-approximate the security for real life attacks.}. The information the program leaks through timing is hence captured by $\algorithm(\traceset)\subseteq \NN$, the image of $\traceset$ under $\algorithm$, and quantified by $\sizeof{\algorithm(\traceset)}\in \NN$.

\subsubsection{Comparing Leaks}

Notions of performance such as relative competitiveness (see Definition~\ref{def:miss-competitive}) are based on trace properties, hence the point of comparison are individual traces. In contrast, information-theoretic notions of leakage are hyperproperties, which makes {\em sets} of traces the natural point of comparison. We now define \emph{leak competitiveness}, a concept that enables us to compare the timing leakage of two cache algorithms, and that is based on lifting miss competitiveness from traces to sets of traces.

\begin{defi}\label{defi:leak_compet}
For a function $\leakrat\colon\NN\rightarrow\RR_{\ge 0}$, we say that algorithm $\algorithm$ is \emph{$\leakrat$-leak-competitive} relative to algorithm $\algorithmprime$ if, for all $\len\in\NN$,
\begin{equation*}
\sizeof{\algorithm(\traceset)}\le  \leakrat(\len)\cdot\sizeof{\algorithmprime(\traceset)},
\end{equation*}
for all set of traces of blocks $\traceset\subseteq\blockset^\len$.
\end{defi}

Even though the definition of leak competitiveness is based on a lifting of miss competitiveness, there are important differences. Most importantly, leak competitiveness of two algorithms $\algorithm, \algorithmprime$ bounds the ratio of leakage for each $\len\in\NN$, whereas miss competitiveness bounds the ratio of hits and misses for all $\len$.
For traces of length $\len$ and an empty initial cache, the number of misses any cache algorithm can produce is in $\{1,\dots,\len\}$, which means that any two algorithms are $\leakrat$-leak-competitive for $\leakrat(\len)=\len$. The question is hence not whether two cache algorithms are leak-competitive, but rather what shape this relationship takes. We introduce the leak ratio to facilitate reasoning about this shape.

 \begin{defi}
  Given a pair of algorithms $\algorithm$ and $\algorithmprime$ we define the \emph{leak ratio} $\leakrat_{\algorithm,\algorithmprime}$ as:
\begin{equation*}
\leakrat_{\algorithm,\algorithmprime}(\len)=\min\{\leakrat(\len)\mid \leakrat\colon\NN\to\RR_{\ge 0}, 
\text{$\algorithm$ is $\leakrat$-leak-competitive relative to $\algorithmprime$}\}.
\end{equation*}
\end{defi}

As we are mostly interested in the asymptotic behavior of $\leakrat_{\algorithm, \algorithmprime}$, the lack of an additive slack in the definition of miss competitiveness is not essential.

\section{Characterizing the Leak Ratio}\label{sec:charleakrat}
In this section we present our main result, which is a
characterization of the asymptotic behavior of the leak ratio for any
pair of cache algorithms. We then give interpretations of this
behavior in terms of security. We present the proofs of the technical
results in Sections~\ref{sec:sequences}-\ref{sec:lowbound}.
\subsection{Non-Dominance}

The key question motivating our work is whether some cache algorithms are preferable to others in terms of their leakage via
timing. This is a natural question to ask because it is well-known
that such preferences relations exist for performance, see
Example~\ref{ex:miss-competitive}. The following theorem gives a
negative answer to the question above.
\begin{thm}\label{thm:disc-nodom}

For each pair of algorithms $\algorithm,\algorithmprime$ we have, as $\len$ grows:
\begin{equation*}
\mathcal{O}(\leakrat_{\algorithm,\algorithmprime}(\len))=\mathcal{O}(\leakrat_{\algorithmprime,\algorithm}(\len)).
\end{equation*}

\end{thm}
Theorem~\ref{thm:disc-nodom} shows that cache algorithms are incomparable in the sense that, for every $\len\in\NN$ and every set of traces $\traceset$ that witnesses an advantage for $\algorithm$ over $\algorithmprime$ in terms of leakage, there is a set of traces $\traceset'$ that witnesses a comparable advantage of $\algorithmprime$ over $\algorithm$.  The following examples exhibits such witnesses for $\algorithm=$ LRU and $\algorithmprime=$ FIFO.

\begin{exa} \label{ex:motivating}
Consider two fully-associative caches of capacity two, one with LRU and the other with FIFO replacement, and the following sets of traces:
\begin{equation*}
\traceset=\left\{\begin{aligned}
&\mathrm{ABACACBBB},\\
&\mathrm{ABACDAAAA},\\
&\mathrm{ABACBADDD},\\
&\mathrm{ABACBACBB},\\
&\mathrm{ABACBACBA}
\end{aligned}\right\}
\quad\quad\quad \traceset'=
\left\{\begin{aligned}
&\mathrm{ABACBAAAA},\\
&\mathrm{ABACDAAAA},\\
&\mathrm{ABACABCCC},\\
&\mathrm{ABACACBCA}
\end{aligned}\right\}
\end{equation*}
Starting from an empty initial cache state, LRU produces 5 different observations on $\traceset$, whereas FIFO produces only one. In contrast, FIFO produces 4 different observations on $\traceset'$ whereas LRU produces only one. 

The root cause for this divergent behavior is that, after accessing the prefix ABAC, the content of both caches differs: block C evicts the least recently used block for LRU (i.e.,~B) but the first block to enter the cache for FIFO (i.e., A). The suffixes of the traces are constructed in such a way that the difference in cache content maps to different observable behavior. The full diagram of updates of the cache when using these sets of traces is given in Figure~\ref{fig:example} in the Appendix. 
\end{exa}
The proof of Theorem~\ref{thm:disc-nodom} is based on a systematic way of constructing sets of traces such as the ones in Example~\ref{ex:motivating}. Formally, the theorem follows from applying Theorem~\ref{thm:nodom}, introduced in Section~\ref{sec:approxsingletrace}, to $\algorithm,\algorithmprime$ and to $\algorithmprime, \algorithm$.  Moreover, we will see later that these sets can be obtained from only two traces of memory blocks and that those two traces are enough to characterize the leak ratio.

\subsection{Shapes of $\leakrat_{\algorithm,\algorithmprime}$}

In Section~\ref{sec:leakrat} we have already observed that the leak
ratio $\leakrat_{\algorithm,\algorithmprime}$ between any two algorithms
$\algorithm$ and $\algorithmprime$ is upper bounded by a linear
function. The interesting question is hence what sublinear shapes
$\leakrat_{\algorithm,\algorithmprime}$ can take. We answer this question
for cache algorithms with finite sets of control states, which
encompasses most hardware-based cache implementations. For this
important class, the following theorem shows that the leak ratio
is either asymptotically constant or linear, ruling out any
nontrivial sublinear shape.

\begin{thm}
\label{thm:linear}
For each pair of algorithms $\algorithm,\algorithmprime$ with finite control
we have either 
\begin{itemize}
\item $\leakrat_{\algorithm,\algorithmprime}(\len)\in\Theta(\len)$, or
\item $\leakrat_{\algorithm,\algorithmprime}(\len)\in\Theta(1)$.
\end{itemize}
\end{thm}
These results are a direct consequence of
Theorem~\ref{thm:lineargrowth}, introduced in Section~\ref{sec:lowbound}, which shows that the leak ratio of two
finite-control algorithms $\algorithm,\algorithmprime$ is lower bounded by a
linear factor if and only if there exist traces that witness that
the difference in misses between $\algorithm$ and $\algorithmprime$ is
unbounded. Whether such traces exist determines in which of the two
classes described by Theorem~\ref{thm:linear} the algorithms $\algorithm$
and $\algorithmprime$ fall. If they do not exist, note that
Corollary~\ref{cor:singletrace} implies that
$\leakrat_{\algorithm,\algorithmprime}\in\mathcal{O}(1)$.

For example, any pair of algorithms with different capacities falls
into the first class. This is because one algorithm always contains a
block that the other does not, which allows to construct a trace of
unbounded difference in misses.

Together with Theorem~\ref{thm:disc-nodom}, Theorem~\ref{thm:linear}
leads to a stronger non-dominance result for finite-control
algorithms, namely that for every $\len\in\NN$ there are sets of traces
$\traceset_\algorithm^\len,\traceset_\algorithmprime^\len\subseteq\blockset^\len$ such that one algorithm asymptotically leaks
the largest possible amount of information whereas the other leaks
almost nothing. That is, $\algorithm(\traceset_\algorithm^\len)\in \Theta(1)$ and
$\algorithmprime(\traceset_\algorithm^\len)\in\Theta(\len)$, whereas
$\algorithm(\traceset_\algorithmprime^\len)\in\Theta(\len)$ and
$\algorithmprime(\traceset_\algorithmprime^\len)\in\Theta(1)$. 

\smallskip 
The non-dominance results from Theorems~\ref{thm:disc-nodom} and \ref{thm:linear} also show why we cannot define leak competitiveness in the same way as miss-competitiveness, that is, where each pair of cache algorithms has a constant leak ratio for all lengths of traces. Except for the case where both leak ratios are in $\Theta(1)$, if we can find, for each length, sets of traces where one cache algorithm leaks more and more information as we increase the length whereas the other leaks a constant amount, no constant value of the leak ratio satisfies the leak-competitiveness definition for all lengths.

We now compute the leak ratio functions for two pairs of cache algorithms to showcase how the constants ignored in the asymptotic results in Theorems~\ref{thm:disc-nodom} and \ref{thm:linear} can mean a small advantage of one cache algorithm over the other, provided the leak ratios are not in $\Theta(1)$.
\begin{exa}
\label{exa:lin-const}
We can compute the leak ratios for small lengths of traces by computing all traces of hits and misses of a given length and then choosing the subset of traces that produces the largest ratio in the number of observations in favor of each algorithm.

To compute these traces of hits and misses we simulate them by exhaustively enumerating all possible traces of memory blocks $\blockset^l$. We now argue the size of the set $\blockset$ needed for the case of capacity two.
For every pair of configurations updated from the initial by the same trace of memory blocks, the last accessed block is cached \eqref{eq:update}, accessing this block again produces a hit for both configurations. The other line of the configuration may store a different memory block for each cache algorithm so that accessing one of them produces a hit for one algorithm and a miss for the other and vice versa. Finally, accessing any memory block not cached for any of the cache algorithms produces a miss for both.
Then, for cache algorithms with capacity two, four memory blocks are enough to simulate all possible traces of hits and misses.

Consider two pairs of fully-associative caches of capacity two, one pair considers replacements LRU and FIFO while the other considers LRU and a cache algorithm that we denote (F)LRU that starts behaving like FIFO but, after seven accesses to memory, behaves like LRU for the remaining of the accesses.
The leak ratios for LRU and FIFO are shown in Figure~\hyperref[fig:linear]{1A} and the ones for LRU and (F)LRU are shown in Figure~\hyperref[fig:constant]{1B}.

We see that the leak ratios of LRU and FIFO exemplify the first case of Theorem~\ref{thm:linear} and that, by looking at the slopes, that the asymptotic approach ignores, we conclude that FIFO has a small advantage over LRU since the leak ratio of FIFO relative to LRU grows slower than that of LRU relative to FIFO, Figure~\hyperref[fig:linear]{1A}.

On the other hand, the leak ratios of LRU and (F)LRU exemplify the second case of Theorem~\ref{thm:linear}. The leak ratios for both cache algorithms start growing at different rates but, once both algorithms behave like LRU the leak ratios end up coinciding and become constant functions.
Although both cache algorithms behave the same starting from length eight, not all pairs of configurations contain the same memory blocks at this point, which still allows for both leak ratios to grow with the length. Once all traces update the pairs of configurations to having the same blocks for both algorithms, the leak ratios become constant functions.

A pair of cache algorithms that are the same or eventually become the same are the only ones that verify the second case of Theorem~\ref{thm:linear}.
This is because the leak ratios grow when the pairs of configurations do not have the same memory blocks cached and the access to a specific memory block produces a hit for one algorithm and a miss for the other. If, at some point, every access to a memory block has the same effect for both algorithms, the leak ratios do not grow anymore.
\begin{figure}
\centering

\begin{subfigure}{0.49 \textwidth}
\centering
\begin{tikzpicture}
\begin{axis}
	[ymin=0.5,
	ymax=12,
	xmin=1,
	xmax=17,
	xlabel= Length of trace,
	ylabel= Leak ratio,
	legend columns=1,
	legend entries={$\leakrat_{\text{LRU,FIFO}}(l)$,$\leakrat_{\text{FIFO,LRU}}(l)$},
	legend pos = north west,
	width=\textwidth,
	height=35 ex]
\addplot[color=red,style=solid] coordinates{
	(1, 1) (2, 1) (3, 1) (4, 1) (5, 2) (6, 3) (7, 4) (8, 4)  (9, 5) (10, 6)
	(11, 7) (12, 7) (13, 8) (14, 9) (15, 10) (16, 10) (17, 11)};
\addplot[color=blue,style=dashed] coordinates{
	(1, 1) (2, 1) (3, 1) (4, 1) (5, 2) (6, 3) (7, 3) (8, 4) (9, 4) (10, 5)
	(11, 5) (12, 6) (13, 6) (14, 7) (15, 7) (16, 8) (17, 8)};
\end{axis}
\end{tikzpicture}
\caption{Comparison of the leak ratios of LRU relative to FIFO and vice versa.}
\label{fig:linear}
\end{subfigure}
~\hfill
\begin{subfigure}{0.49\textwidth}
\centering
\begin{tikzpicture}
\begin{axis}
	[ymin=0.5,
	ymax=12,
	xmin=1,
	xmax=17,
	width=\textwidth,
	xlabel= Length of trace,
	ylabel= Leak ratio,
	legend columns=1,
	legend entries={$\leakrat_{\text{LRU,(F)LRU}}(l)$,$\leakrat_{\text{(F)LRU,LRU}}(l)$},
	legend pos = north west,
	height=35 ex]
\addplot[color=red,style=solid] coordinates{
	(1, 1) (2, 1) (3, 1) (4, 1) (5, 2) (6, 3) (7, 4) (8, 4) (9, 5) (10, 6)
	(11, 6) (12, 6) (13, 6) (14, 6) (15, 6) (16, 6) (17, 6)};
\addplot[color=blue,style=dashed] coordinates{
	(1, 1) (2, 1) (3, 1) (4, 1) (5, 2) (6, 3) (7, 3) (8, 4) (9, 4) (10, 5)
	(11, 5) (12, 6) (13, 6) (14, 6) (15, 6) (16, 6) (17, 6)};
\end{axis}
\end{tikzpicture}
\caption{Comparison of the leak ratios of LRU relative to (F)LRU and vice versa.}
\label{fig:constant}
\end{subfigure}
\caption{Example of the behavior of the leak ratios of the cache algorithms from Example~\ref{exa:lin-const}. Figure~\hyperref[fig:linear]{1A} shows two leak ratios that grow asymptotically linearly and where the slopes can give a slight advantage of one algorithm over the other. Figure~\hyperref[fig:constant]{1A} shows two leak ratios that eventually become constant functions and were, for large lengths of traces, there is no advantage of one algorithm over the other.}
\label{fig:ratios}
\end{figure}
\end{exa}
By using a constant notion of leak competitiveness, all cache algorithms, except for the ones in $\Theta(1)$, would be deemed incomparable. On the other hand, by defining the leak ratios as functions of the length of the trace and observing the growth rate as in Example~\ref{exa:lin-const}, we can establish a comparison between cache algorithms.

\subsection{Discussion}\label{sec:discussion} 

We now discuss the implications of Theorems~\ref{thm:disc-nodom} and~\ref{thm:linear} (short: {\em our results}) in practice.

\begin{enumerate}
\item Our results are asymptotic in nature. The constants hidden behind the $\mathcal{O}$-notation can indicate a (gradual) preference between algorithms on finite sets of traces. E.g., the traces in Example~\ref{ex:motivating} and the different slopes on Example~\ref{exa:lin-const}  show a slight advantage of FIFO over LRU.
\item Our results rely on the construction of sets of traces that witness advantages of one cache algorithm over another, see Example~\ref{ex:motivating}. However, the constructed traces need not correspond to a program of interest. Restricting to a specific class of programs corresponds to the constraint that witnesses be picked from a subset $\traceset\subseteq\blockset^\len$ instead of~$\blockset^\len$. Under such constraints, it may be possible that a preference relation between cache algorithms exists.
\item Our results rely on the assumptions that the caching algorithm is deterministic and based on demand paging, i.e., it loads blocks only when they are requested by the program. It is possible that randomized policies or features such as prefetching enable one to sidestep our results. For example, for miss-competitiveness it is known that randomized policies~\cite{fiat1991} achieve better bounds than those possible for deterministic policies~\cite{sleator1985amortized}.
The study of leak competitiveness for randomized cache algorithms is out of the scope of this paper. We are aware of non-dominance results for a similar notion of leak-competitiveness that consider the RANDOM cache algorithm that, upon a miss, evicts a memory block randomly~\cite{schroder2018}.

\item Our results rely on an adversary that can observe the overall execution time of the program as in, e.g.~\cite{Bernstein05cache-timingattacks,osvikshamir06cache}. They do not necessarily hold for adversaries that can observe the cache state after or during the computation of the victim. Such attacks are possible whenever the adversary shares the cache with the victim, which has shown to be the most effective attack vector. In contrast, our results are relevant for remote attacks, which are less effective, but harder to detect and defend against. We briefly discuss the case of access-based adversaries in Section~\ref{sec:related}.  
\end{enumerate}
Despite these limitations in scope, we do believe that our results lay an interesting basis for theory research in the domain of microarchitectural side-channel attacks, where foundational results are still scarce.

\section{Leak Ratio from a Pair of Traces}
\label{sec:sequences}
Information leakage is a hyperproperty, i.e., a property of sets of traces.
We now show that the leak can always be
expressed in terms of the difference in observations of only two
traces of memory blocks.  

For an algorithm $\algorithmprime$ and $\len \in \NN$, we
say that $\trace_1,\trace_2\in \blockset^\len$ are {\em $\algorithmprime$-equivalent}
whenever $\algorithmprime(\trace_1)=\algorithmprime(\trace_2)$. We say that a set $\traceset\subseteq\blockset^\len$
is {\em $\algorithmprime$-dense} if the image of $\traceset$ under $\algorithmprime$
is a contiguous sequence of natural numbers,
i.e. $\algorithmprime(\traceset)=\{j,j+1,\dots,j+k\}$ for some $j,k\in\NN$.

\begin{prop}
\label{prop:maxset}
For all pairs of algorithms $\algorithm$ and $\algorithmprime$, all
lengths $\len$, and all pairs of $\algorithmprime$-equivalent traces of memory blocks
$\trace_1, \trace_2\in\blockset^\len$:
\begin{equation}\label{eq:sequences}
\algorithm(\trace_2)-\algorithm(\trace_1) \le \leakrat_{\algorithm,\algorithmprime}(\len)-1.
\end{equation}
Moreover, there exist pairs of traces of $\algorithmprime$-equivalent
memory blocks $\trace_1, \trace_2\in\blockset^\len$ such that~\eqref{eq:sequences} is
an equality.
\end{prop}

That is, every pair of traces that coincides in timing observation on one
algorithm cannot differ by more than the leak ratio on the other
algorithm. Moreover, there exists a pair of traces that matches this bound.

The proof of the upper bound is based on constructing a set
$\traceset\subseteq\blockset^\len$ of $\algorithmprime$-equivalent traces from a pair
$\trace, \trace'\in\blockset^\len$ of $\algorithmprime$-equivalent traces. The set
$\traceset$ is $\algorithm$-dense with maximum $\algorithm(\trace')$ and minimum
$\algorithm(\trace)$. It satisfies
\begin{equation}\label{eq:setquotient}
\leakrat_{\algorithm,\algorithmprime}(\len)\ge\frac{\sizeof{\algorithm(\traceset)}}{\sizeof{\algorithmprime(\traceset)}},
\end{equation}
which equals $\algorithm(\trace)-\algorithm(\trace')+1$ by construction.

The following lemma describes the construction of traces $\algorithmprime$-equivalent to
$\trace$ and $\trace'$ whose number of misses for $\algorithm$ cover every value in between
$\algorithm(\trace)$ and $\algorithm(\trace')$.
The set $\traceset$ is composed of these traces.
\begin{lem}
\label{lemm:set_seq}
Consider two $\algorithmprime$-equivalent traces $\trace$, $\trace'\in\blockset^\len$ with
$\algorithm(\trace)\leq\algorithm(\trace')$.  Then, for
every $\algorithm(\trace)\leq k\leq \algorithm(\trace')$ there exists a trace
$\trace^*\in\blockset^\len$ such that $\algorithm(\trace^*)=k$ and that is
$\algorithmprime$-equivalent to $\trace$ and $\trace'$.
\end{lem}

\begin{proof}
  We begin with a continuity argument to identify a prefix of the
  trace $\trace$, which we later extend to $\trace^*$.
  For this, note that the difference in misses, $\algorithmprime-\algorithm$,
  between both algorithms on trace $\trace$ is initially zero,
  i.e. $\algorithmprime(\epsilon)-\algorithm(\epsilon)=0$, and increases or
  decreases by at most $1$ per added block, until it reaches
  $\algorithmprime(\trace)-\algorithm(\trace)$. We first consider the case
  $\algorithmprime(\trace)\ge\algorithm(\trace)$. For any $k$ with
  $0 \le \algorithmprime(\trace)-k\le \algorithmprime(\trace)-\algorithm(\trace)$, we hence
  find a prefix $\block_1\cdots \block_u$ of $\trace$ such that the value of
  $\algorithmprime-\algorithm$ on the prefix is exactly $\algorithmprime(\trace)-k$:
  \begin{equation}\label{eq:diff}
    \algorithmprime(\block_1\cdots \block_u)-\algorithm(\block_1\cdots \block_u)=\algorithmprime(\trace)-k.
  \end{equation}

We create a trace $\trace^*$ with prefix
$\block^*_1\dots \block_u^*=\block_1\dots \block_u$, which we extend by blocks $\block^*_{u+1}\dots \block^*_v$ that produce misses on both $\algorithm$ and
$\algorithmprime$ until
\begin{equation}\label{eq:pref}
\algorithmprime(\block_1^*\dots \block^*_v)=\algorithmprime(\trace)\ .
\end{equation}
For the blocks $\block^*_{u+1}\dots \block^*_v$ to miss they must be uncached in
both $\algorithm$ and $\algorithmprime$; such blocks can be found whenever
$\blockset$ is larger than the sum of the the capacities of both
algorithms.  We further extend $\block_1^*\dots \block^*_v$ with $\len-v$ copies
of $\block_v^*$ to the trace $\trace^*$ of length $\len$. Repeatedly accessing
$\block_v^*$ is guaranteed to produce hits on both $\algorithm$ and
$\algorithmprime$.

As the blocks $\block^*_{u+1}\dots \block^*_{\len}$ produce identical outputs on $\algorithm$
and $\algorithmprime$, the trace $\trace^*$ still satisfies~\eqref{eq:diff}, i.e.,
\begin{equation*}
\algorithmprime(\trace^*)-\algorithm(\trace^*)=\algorithmprime(\trace)-k\ .
\end{equation*}
Moreover, $\trace^*$ also still satisfies~\eqref{eq:pref}, i.e.,
$\algorithmprime(\trace^*)=\algorithmprime(\trace)$, from which we conclude that
$\trace^*$ is $\algorithmprime$-equivalent to $\trace$ and
$\algorithm(\trace^*)=k$.
Note that we only handled the case $\algorithm(\trace)\le\algorithmprime(\trace)$ so that $k\leq \algorithmprime(\trace)$.
The case $\algorithm(\trace)>\algorithmprime(\trace)$ where $k> \algorithmprime(\trace)$ proceeds in the same way but
extending a prefix of $\trace'$ instead of $\trace$ and reformulating~\eqref{eq:diff} to 
$\algorithm(\block'_1\cdots \block'_u) - \algorithmprime(\block'_1\cdots \block'_u)=k-\algorithmprime(\trace')$.
\end{proof}

\begin{exa}
  Consider the cache algorithms $\algorithm=$ LRU and $\algorithmprime=$ FIFO and the traces of memory blocks ABACACBBB and ABACBACBA, as in Example~\ref{ex:motivating}.  Both traces are FIFO-equivalent, however LRU$(\mathrm{ABACACBBB})=4$ and LRU$(\mathrm{ABACBACBA})=8$.  Then, following Lemma~\ref{lemm:set_seq}, there exist three traces of memory blocks that are FIFO-equivalent but where LRU produces between 5 and 7 misses, namely:
\begin{equation*}
\{\mathrm{ABACDAAAA},
\mathrm{ABACBADDD},
\mathrm{ABACBACBB}\}
\end{equation*}
The union of this set with the two initial traces yields the set $\traceset$ from Example~\ref{ex:motivating}.
\end{exa}

The proof of the tightness of the upper bound in
Proposition~\ref{prop:maxset} follows from the fact that every set $\traceset$
that satisfies equality in \eqref{eq:setquotient} contains within it a
subset $\traceset^*$ of $\algorithmprime$-equivalent traces that also satisfies
equality in~\eqref{eq:setquotient}. We show that this set $\traceset^*$ is
$\algorithm$-dense, which means that the elements $\trace,\trace'\in \traceset^*$ that
produce the maximal difference in misses under $\algorithm$ satisfy
$\algorithm(\trace)-\algorithm(\trace')=\leakrat_{\algorithm,\algorithmprime}-1$.

The following lemma shows how to find such a $\traceset^*$.
\begin{lem}
  Every set $\traceset\subseteq \blockset^\len$ that satisfies equality in
  \eqref{eq:setquotient} contains a $\algorithm$-dense subset of
  $\algorithmprime$-equivalent traces  that also satisfies equality in
  \eqref{eq:setquotient}.
\end{lem}
\begin{proof}
  We partition $\traceset=\traceset_1\uplus\dots\uplus \traceset_k,$ into classes
  of $\algorithmprime$-equivalent traces. Without loss of generality
  assume that $\algorithm(\traceset_1)\geq \algorithm(\traceset_j)$, for $j>1$. Then we
  have:
\begin{equation*}
  \frac{\sizeof{\algorithm(\traceset)}}{\sizeof{\algorithmprime(\traceset)}}
  \leq \frac{\sum_{j=1}^k
    \sizeof{\algorithm(\traceset_j)}}{\sum_{j=1}^k
    \sizeof{\algorithmprime(\traceset_j)}} = \frac{\sum_{j=1}^k
    \sizeof{\algorithm(\traceset_j)}}{\sum_{j=1}^k
    1} \stackrel{(*)}{\leq}\sizeof{\algorithm(\traceset_1)}\ ,
\end{equation*}
where $(*)$ follows from the fact that, for any sequence of natural numbers $a_1,\dots,a_k$, $\sum_{j=1}^k a_j \leq k\max(a_1,\dots, a_k)$.
As a consequence, $\traceset_1$ also satisfies
$\sizeof{\algorithm(\traceset_1)}=\leakrat_{\algorithm,\algorithmprime}(\len)$. Moreover,
$\algorithm(\traceset_1)$ is a contiguous set of natural numbers. If it were not,
we could apply Lemma~\ref{lemm:set_seq} to augment $\traceset_1$ by a trace that
produces the missing number of observations, contradicting that
$\leakrat_{\algorithm,\algorithmprime}$ is an upper bound.
\end{proof}

\section{Approximation of the Leak Ratio from a Single Trace}\label{sec:approxsingletrace}
In Section~\ref{sec:sequences} we have seen that the leak ratio of two cache algorithms, which is defined as a property of arbitrary sets of traces, is fully characterized by a pair of traces.
In this section, we show that the leak ratio can be approximated to within a factor of 2 using a single trace.

\begin{lem}\label{lem:singletrace}
	Let $\trace \in \blockset^\len$ be an arbitrary trace.
	Then, there is a trace $\trace' \in \blockset^\len$ with $\algorithm(\trace') = \algorithmprime(\trace') = \algorithmprime(\trace)$.
\end{lem}
\begin{proof}
		We construct the trace $\trace' \in \blockset^\len$ as the concatenation of two subtraces $\trace'_\textit{miss}$ and $\trace'_\textit{hit}$:
	$\trace'_\textit{miss}$ is a trace of length $\algorithmprime(\trace)$ in which all accesses are chosen such that they result in misses in both $\algorithm$ and $\algorithmprime$.
		This is always possible, as there are at most $\capacity_\algorithm+\capacity_\algorithmprime$ blocks cached in $\algorithm$ and $\algorithmprime$ at any time and accesses to any other block will result in a miss.
		Let $\block \in \blockset$ be the final access in $\trace'_\textit{miss}$.
		Independently of the cache algorithm, $\block$ must be cached in both $\algorithm$ and $\algorithmprime$ following $\trace'_\textit{miss}$.
		The second subtrace $\trace'_\textit{hit}$ then simply consists of $|\trace|-\algorithmprime(\trace)$ accesses to $\block$, which will result in hits in both $\algorithm$ and $\algorithmprime$.
\end{proof}

The following corollary of the previous lemma and of Proposition~\ref{prop:maxset} shows that the leakage ratio is ``almost'' a trace property, as it can be approximated to within a factor of two based on the number of misses of $\algorithm$ and $\algorithmprime$ on a single trace:
\begin{cor}\label{cor:singletrace}
	For all pairs of cache algorithms $\algorithm$ and $\algorithmprime$, all lengths $\len$, and all traces of memory blocks $\trace \in \blockset^\len$:
	\begin{equation}\label{eq:traceprop}
			|\algorithm(\trace)-\algorithmprime(\trace)| \leq \leakrat_{\algorithm,\algorithmprime}(\len)-1.
	\end{equation}
	Moreover, there exists a trace $\trace \in \blockset^\len$ such that: \[\frac{\leakrat_{\algorithm,\algorithmprime}(\len)-1}{2} \leq |\algorithm(\trace)-\algorithmprime(\trace)|.\]
\end{cor}
\begin{proof}
	Let $	\trace \in \blockset^\len$ be an arbitrary trace.
	By Lemma~\ref{lem:singletrace}, there is a trace $\trace'$ such that $\algorithm(\trace')=\algorithmprime(\trace')=\algorithmprime(\trace)$.
	So $\trace$ and $\trace'$ are \algorithmprime-equivalent.
	Thus, by Proposition~\ref{prop:maxset}, we have both
	\begin{equation*}
		\algorithm(\trace) - \algorithm(\trace')  \leq \leakrat_{\algorithm,\algorithmprime}(\len)-1 \quad
                                         \textit{ and }\quad \algorithm(\trace') - \algorithm(\trace)  \leq \leakrat_{\algorithm,\algorithmprime}(\len)-1,	
	\end{equation*}
	which implies that $|\algorithm(\trace)-\algorithmprime(\trace)|  = |\algorithm(\trace)-\algorithm(\trace')|  \leq \leakrat_{\algorithm,\algorithmprime}(\len)-1$.
	
	By Proposition~\ref{prop:maxset}, there is a pair of \algorithmprime-equivalent traces $\trace_1, \trace_2 \in \blockset^\len$ such that:
	\[\algorithm(\trace_2) - \algorithm(\trace_1) = \leakrat_{\algorithm,\algorithmprime}(\len)-1.\]
	Let $q = \algorithmprime(\trace_1) = \algorithmprime(\trace_2)$.
	Either $2\cdot |\algorithm(\trace_2) - q| \geq {\algorithm(\trace_2)-\algorithm(\trace_1)}$ or $2\cdot |\algorithm(\trace_1) - q| \geq {\algorithm(\trace_2)-\algorithm(\trace_1)}$, where equality is achieved on one of the two inequalities if $q$ is centered between $\algorithm(\trace_1)$ and $\algorithm(\trace_2)$.
	Assume that $2\cdot |\algorithm(\trace_2) - q| \geq {\algorithm(\trace_2)-\algorithm(\trace_1)}$.
	Then $|\algorithm(\trace_2) - \algorithmprime(\trace_2)| \geq \frac{\algorithm(\trace_2)-\algorithm(\trace_1)}{2} = \frac{\leakrat_{\algorithm,\algorithmprime}-1}{2}$.
	Otherwise, $|\algorithm(\trace_1) - \algorithmprime(\trace_1)| \geq \frac{\algorithm(\trace_2)-\algorithm(\trace_1)}{2} = \frac{\leakrat_{\algorithm,\algorithmprime}-1}{2}$.
\end{proof}

\begin{thm}\label{thm:nodom}
	For all pairs of cache algorithms $\algorithm$ and $\algorithmprime$ and all lengths $\len$:
		\[\leakrat_{\algorithm,\algorithmprime}(\len) \leq 2\cdot \leakrat_{\algorithmprime,\algorithm}(\len)-1\]
\end{thm}
\begin{proof}
	By Corollary~\ref{cor:singletrace}, there is a trace $\trace \in \blockset^\len$, such that
	\[\frac{\leakrat_{\algorithmprime,\algorithm}(\len)-1}{2} \leq |\algorithmprime(\trace) - \algorithm(\trace)| = |\algorithm(\trace) - \algorithmprime(\trace)| \leq \leakrat_{\algorithm,\algorithmprime}(\len)-1.\]
	Multiplying both sides by $2$ and adding $1$ finish the proof.
\end{proof}

\section{A Linear Lower Bound on the Leak Ratio}\label{sec:lowbound}
In this section, we show that if the difference in misses between two
cache algorithms is unbounded, then there
are traces on which the difference in misses grows linearly in the
length of the trace. Together with the result from the previous
section, this implies that the leak ratio between two algorithms grows
linearly in the length of the trace if and only if the difference
between the two algorithms is unbounded.
This result does not hold for arbitrary caches conforming to the model introduced in Section~\ref{sec:prelim}.
We need to make two additional assumptions:
\begin{enumerate}
	\item We assume the set of control states $\controlstateset_\algorithm$ of a cache algorithm to be finite. This is naturally the case for hardware-based caches that maintain a finite set of status bits to guide future eviction decisions.
	\item We assume that the \emph{evict} function,  $\evict_\algorithm : \controlstateset_\algorithm \times \blockset \rightarrow \controlstateset_\algorithm \times \{0, \dots, \capacity_\algorithm-1\}$ is independent of its second parameter, i.e., $\evict_\algorithm(\controlstate, \block) = \evict_\algorithm(\controlstate,\block')$ for all $\controlstate \in \controlstateset_\algorithm$ and $\block, \block' \in \blockset$.
		This assumption is naturally fulfilled by fully-associative caches, where there is no restriction on the placement of a memory block based on its address.
		This assumption could be significantly weakened at the expense of a more complicated proof.\footnote{A weaker, yet sufficient condition would be that there is a finite partition of $\blockset$, such that $\evict_\algorithm(\controlstate, \block) = \evict_\algorithm(\controlstate,\block')$ for all $\controlstate \in \controlstateset_\algorithm$ and $\block, \block'$ that are in the same block of the partition. This weaker assumption is fulfilled by arbitrary set-associative caches.}
\end{enumerate}

For the proof of the result we argue that, while there is an unbounded
number of different cache configurations, even assuming an unbounded
supply of memory blocks $\blockset$, there are only finitely many
``non-congruent'' pairs of cache configurations, where congruent will
be defined precisely below.  Intuitively, congruent pairs of cache
configurations behave similarly to each other, if their cache contents
are appropriately renamed.

Such a renaming can be captured by a bijection.
Let $\pi : \blockset \rightarrow \blockset$ be a bijection on memory blocks and let $\pi^*$ denote its extension to cache contents that maps $\invalid$ to $\invalid$:
\[\pi^*(\content) = \lambda \len.\begin{cases}
						\pi(\content(\len))	& : \content(\len) \in \blockset\\
						\invalid		& : \content(\len) = \invalid
					\end{cases}\]
We also lift $\pi$ to cache configurations with $\pi^*(\controlstate,\content) = (\controlstate, \pi^*(\content))$ and to traces with $\pi^*(\epsilon) = \epsilon$ and $\pi^*(\block \circ \trace) = \pi(\block) \circ \pi^*(\trace)$.

Let $(\controlstate,\content)$ be an arbitrary cache configuration.
Observe that:
\begin{equation}\label{eq:congruentsuccessors}
	\forall \trace \in \blockset^*: \pi^*(\upconfig_\algorithm((\controlstate,\content), \trace)) = \upconfig_\algorithm(\pi^*(\controlstate,\content), \pi^*(\trace)),
\end{equation}
i.e. renamed cache configurations behave the same on renamed accesses. Also observe that:
\begin{equation}
	\miss_\algorithm(\block, \content) = \miss_\algorithm(\pi(\block), \pi^*(\content)),
\end{equation}
which holds because $\pi(\block)$ is contained in $\pi^*(\content)$ if and only if $\block$ is contained in $\content$.
From these two observations, it follows that:
\begin{equation}\label{eq:congruentmisses}
	\algorithm((\controlstate,\content), \trace) = \algorithm(\pi^*(\controlstate,\content), \pi(\trace)).
\end{equation}

\begin{defi}[Congruent cache configurations]
	Two pairs of cache configurations $(\config_\algorithm, \config_\algorithmprime)$ and $(\config_\algorithm', \config_\algorithmprime')$ are \emph{congruent}, denoted by $(\config_\algorithm, \config_\algorithmprime)) \equiv (\config_\algorithm', \config_\algorithmprime')$, if there is a bijection $\pi : \blockset \rightarrow \blockset$, such that $\config_\algorithm' = \pi^*(\config_\algorithm)$ and $\config_\algorithmprime'= \pi^*(\config_\algorithmprime)$.
	To indicate a bijection $\pi$ that is a witness to the congruence of two pairs of cache configurations we also write $(\config_\algorithm,\config_\algorithmprime) \equiv_\pi (\config_\algorithm',\config_\algorithmprime')$.
	
Note that congruence is an equivalence relation.
We denote the equivalence class of a pair of cache configuration $(\config_\algorithm,\config_\algorithmprime)$ by \[[\config_\algorithm,\config_\algorithmprime] := \{(\config_\algorithm',\config_\algorithmprime') \in \configurations_\algorithm \times \configurations_\algorithmprime \mid (\config_\algorithm',\config_\algorithmprime') \equiv (\config_\algorithm,\config_\algorithmprime)\}.\]
\end{defi}

While the set of pairs of cache configurations is infinite, its quotient w.r.t. to the congruence relation is finite:
\begin{thm}[Index of $\equiv$]\label{thm:finiteindex}
	Let $\algorithm$ and $\algorithmprime$ be two finite-control-state cache algorithms.
	Then, the quotient 
	\[\configurations_\algorithm\times\configurations_\algorithmprime/_\equiv\, = \{ [\config_\algorithm,\config_\algorithmprime] \mid (\config_\algorithm,\config_\algorithmprime) \in \configurations_\algorithm \times \configurations_\algorithmprime \}\]
	is finite.
\end{thm}
\begin{proof}
\newcommand{\bpq}{B_{\algorithm,\algorithmprime}}
	Remember that $\capacity_\algorithm$ and $\capacity_\algorithmprime$ denote the capacities of $\algorithm$ and $\algorithmprime$.
	Let $\bpq$ be an arbitrary but fixed subset of $\blockset$, such that $|\bpq| = \capacity_\algorithm+\capacity_\algorithmprime$.
	
	We show below that each pair $((\controlstate_\algorithm,\content_\algorithm),(\controlstate_\algorithmprime,\content_\algorithmprime))$ of cache configurations is congruent to a pair of cache configurations $((\controlstate_\algorithm,\content_\algorithm'),(\controlstate_\algorithmprime,\content_\algorithmprime'))$ in which only blocks from $\bpq$ may occur in the cache contents $\content_\algorithm'$ and $\content_\algorithmprime'$.
	As $\bpq$ is finite, there are only finitely many different cache contents $\content_\algorithm'$ and $\content_\algorithmprime'$ containing only blocks from $\bpq$.
	The sets of control states $\controlstateset_\algorithm$ and $\controlstateset_\algorithmprime$ are finite by assumption.
	Together, this implies that the set of equivalence classes of $\equiv$ is finite.	
	
	Below we show how to incrementally construct a bijection $\pi : \blockset \rightarrow \blockset$ such that the contents of $\content_\algorithm' = \pi^*(\content_\algorithm)$ and $\content_\algorithmprime' = \pi^*(\content_\algorithmprime)$ contain only blocks from $\bpq$:
	\begin{enumerate}
		\item Initially, let $\pi$ be the identity function on $\blockset$, and let $D = \bpq$.
		\item For $i = 0, \dots, \capacity_\algorithm-1$:\\
		If $\content_\algorithm(i) \in \bpq$ then modify $D$ to $D = D \setminus \{\content_\algorithm(i)\}$.
		\item For $j= 0, \dots, \capacity_\algorithmprime-1$:\\
		If $\content_\algorithmprime(j) \in \bpq$ then modify $D$ to $D = D \setminus \{\content_\algorithmprime(j)\}$.
		\item For $i = 0, \dots, \capacity_\algorithm-1$:\\
			If $\content_\algorithm(i) \not= \invalid$ and $\pi(\content_\algorithm(i)) \not\in \bpq$, then pick $\block \in D$ and modify $\pi$ and $D$ as follows:\\
				\hspace{5cm}$\quad\quad\pi = \pi[\content_\algorithm(i) \mapsto \block][\block \mapsto \content_\algorithm(i)]$ and $D = D \setminus \{\block\}$.
		\item For $j = 0, \dots, \capacity_\algorithmprime-1$:\\
			If $\content_\algorithmprime(j) \not= \invalid$ and $\pi(\content_\algorithmprime(j)) \not\in \bpq$, then pick $\block \in D$ and modify $\pi$ and $D$ as follows:\\
				\hspace{5cm}$\quad\quad\pi = \pi[\content_\algorithmprime(j) \mapsto \block][\block \mapsto \content_\algorithmprime(j)]$ and $D = D \setminus \{\block\}$.
	\end{enumerate}
	Note that there is always a $\block \in D$ available, when the above algorithm needs one, because the operation is applied at most $|\bpq| = \capacity_\algorithm + \capacity_\algorithmprime$ times.
	Throughout its execution, the algorithm maintains the invariant that $\pi$ is a bijection.
	Further, the resulting bijection satisfies $\pi^*(\content_\algorithm), \pi^*(\content_\algorithmprime) \subseteq \bpq \cup \{\invalid\}$.
\end{proof}

We can exploit Theorem~\ref{thm:finiteindex} in a manner similar to the application of the pumping lemma for regular languages in the proof of the following theorem.
\begin{thm}\label{thm:lineardiff}
	Let $\algorithm$ and $\algorithmprime$ be two finite-control-state cache algorithms.
	Further, let the difference in misses between $\algorithm$ and $\algorithmprime$ be unbounded, i.e.,
	 \[\forall m \in \NN: \exists \trace \in \blockset^*: |\algorithm(\trace) - \algorithmprime(\trace)| > m.\]
	 Then, there is an $f \in \RR, f > 0$ and an $m_0 \in \NN$, such that
	 \[\forall m \in \NN, m > m_0: \exists \trace \in \blockset^m: |\algorithm(\trace) - \algorithmprime(\trace)| > f\cdot |\trace|.\]
\end{thm}
\begin{proof}
	Let $\algorithm$ and $\algorithmprime$ be two finite-control-state cache algorithms such that the difference in misses between $\algorithm$ and $\algorithmprime$ is unbounded.
	Let $\len = |\configurations_\algorithm\times\configurations_\algorithmprime/_\equiv|+1$, which must be finite due to Theorem~\ref{thm:finiteindex}.
		 
	As the difference in misses between $\algorithm$ and $\algorithmprime$ is unbounded, there must be a $\trace \in \blockset^*$ such that $|\algorithm(\trace) - \algorithmprime(\trace)| = \len$.
	We will assume without loss of generality\footnote{If $\algorithm(\trace) < \algorithmprime(\trace)$ the following arguments hold with $\algorithm$ and $\algorithmprime$ exchanged.} that $\algorithm(\trace) > \algorithmprime(\trace)$ for such traces $\trace$, and so $|\algorithm(\trace) - \algorithmprime(\trace)| = \algorithm(\trace) - \algorithmprime(\trace)$.
	Then, let $\trace_1, \dots, \trace_\len$ be prefixes of $\trace$, s.t. $\algorithm(\trace_j) - \algorithmprime(\trace_j) = j$ for all $1 \leq j \leq \len$.
	
	Let $\initialconf_\algorithm = (\initialstate_\algorithm, \lambda j.\invalid)$ and $\initialconf_\algorithmprime = (\initialstate_\algorithmprime, \lambda j.\invalid)$ be the initial configurations of $\algorithm$ and $\algorithmprime$.
	Also, let $p_j = \upconfig_\algorithm(\initialconf_\algorithm, \trace_j)$ and $q_j = \upconfig_\algorithmprime(\initialconf_\algorithmprime, \trace_j)$ for all $1 \leq j \leq \len$.
	
	Due to the pigeonhole principle, there must be at least two prefixes $\trace_j$ and $\trace_k$, with $j < k$, such that the pairs of cache configurations $(p_j, q_j)$ and $(p_k, q_k)$ resulting from executing these prefixes are congruent. 
	Assume that $\trace_j$ and $\trace_k$ are two such prefixes.
	
	\newcommand{\tracejtok}{\trace_{j \rightarrow k}}
	
	As $\trace_j$ is a prefix of $\trace_k$, we can decompose $\trace_k$ into $\trace_j$ and $\tracejtok$, such that $\trace_k = \trace_j \circ \tracejtok$.
	From $\algorithm(\trace_j) - \algorithmprime(\trace_j) = j$ and $\algorithm(\trace_k) - \algorithmprime(\trace_k) = k$ we can conclude that $\algorithm(\trace_j,\tracejtok) - \algorithmprime(\trace_j,\tracejtok) = (\algorithm(\trace_k)-\algorithm(\trace_j)) - (\algorithmprime(\trace_k)-\algorithmprime(\trace_j)) = k-j \geq 1$.
		
	 We can arbitrarily extend $\trace_j$ using the following construction of the traces $\tau_m$ and $\omega_m$:
	 \begin{align*}
	 	\tau_0 & = \trace_j,\\
		\tau_{m+1} & = \tau_m\circ \omega_m,\\
		\omega_0 & = \tracejtok,\\
		\omega_{m+1} & = \pi^*(\omega_m).
	 \end{align*} 	   
	 Let $u_m = \upconfig_\algorithm(\initialconf_\algorithm,\tau_m)$ and $v_m = \upconfig_\algorithmprime(\initialconf_\algorithmprime,\tau_m)$.	 
	 For the following induction proof, it will be helpful to express $u_{m+1}$ and $v_{m+1}$ in terms of $u_m$ and $v_m$. We have that $u_{m+1} = \upconfig_\algorithm(\initialconf_\algorithm,\tau_{m}\circ\omega_{m}) = \upconfig_\algorithm(\upconfig_\algorithm(\initialconf_\algorithm,\tau_{m}),\omega_{m}) = \upconfig_\algorithm(u_{m},\omega_{m})$ and similarly $v_{m+1} = \upconfig_\algorithmprime(v_{m},\omega_{m})$.

	 We can show by induction that $(u_m, v_m) \equiv_\pi (u_{m+1}, v_{m+1})$:
	 \begin{itemize}
	 	\item (Induction base) For $m=0$,  $\tau_0 = \trace_j$ and  $\tau_{1} = \tau_0\circ\omega_0=\trace_j\circ\tracejtok = \trace_k$.
	 Thus we have that $u_0 = \upconfig_\algorithm(\initialconf_\algorithm,\tau_0) = \upconfig_\algorithm(\initialconf_\algorithm,\trace_j) = p_j$ and $v_0 = \upconfig_\algorithmprime(\initialconf_\algorithmprime,\tau_0) = \upconfig_\algorithmprime(\initialconf_\algorithmprime,\trace_j) = q_j$. Similarly, $u_1 = p_k$ and $v_1 = q_k$, and we already know that $(p_j,q_j) \equiv_\pi (p_k, q_k)$.
	 	\item (Induction step) For $m > 0$, we know from the induction hypothesis that $(u_{m-1}, v_{m-1}) \equiv_\pi (u_m, v_m)$.
			Applying (\ref{eq:congruentsuccessors}) with $\trace = \omega_{m-1}$ yields 
				$\pi^*(u_m) = \pi^*(\upconfig_\algorithm(u_{m-1}, \omega_{m-1})) = \upconfig_\algorithm(u_m, \pi^*(\omega_{m-1})) = \upconfig_\algorithm(u_m, \omega_m) = u_{m+1}$, and similarly $\pi^*(v_m) = v_{m+1}$. 
				Thus $(u_m, v_m) \equiv_\pi (u_{m+1}, v_{m+1})$.
	 \end{itemize}
	 
	 Since we have that $u_{m+1} = \pi(u_m)$ and $v_{m+1} = \pi(v_m)$, applying (\ref{eq:congruentmisses}) yields that $\algorithm(u_{m+1}, \omega_{m+1}) = \algorithm(u_{m+1}, \pi^*(\omega_m)) = \algorithm(u_m, \omega_m)$ and similarly we have $\algorithm(v_{m+1}, \omega_{m+1}) = \algorithm(v_m, \omega_m)$ for all $m$. In other words, the number of misses on the subtraces $\omega_m$ are always the same in both $\algorithm$ and $\algorithmprime$.
	 We also know that $\algorithm(u_0, \omega_0) = \algorithm(\trace_j,\tracejtok)$ and $\algorithmprime(v_0, \omega_0) = \algorithmprime(\trace_j,\tracejtok)$. Thus, we have
	 \begin{align*}
	 	\algorithm(\tau_m) & = \algorithm(\trace_j) + m\cdot \algorithm(\trace_j,\tracejtok),\\
	 	\algorithmprime(\tau_m) & = \algorithmprime(\trace_j) + m\cdot \algorithmprime(\trace_j,\tracejtok),\\
	 	\algorithm(\tau_m) - \algorithmprime(\tau_m) & = \algorithm(\trace_j)-\algorithmprime(\trace_j) + m\cdot (\algorithm(\trace_j,\tracejtok)-\algorithmprime(\trace_j,\tracejtok)),\\
						& = \algorithm(\trace_j)-\algorithmprime(\trace_j) + m\cdot (k-j).
	 \end{align*}
	Let $f = \frac{k-j}{|\tracejtok|+1} \geq  \frac{1}{|\tracejtok|+1} > 0$.
	For large enough $m$, $|\algorithm(\tau_m) - \algorithmprime(\tau_m)| = \algorithm(\trace_j)-\algorithmprime(\trace_j) + m\cdot (k-j)$ is greater than $f\cdot |\tau_m| = \frac{k-j}{|\tracejtok|+1}\cdot (|\tau_0|+m\cdot|\tracejtok|)$, which proves the theorem.
\end{proof}

In other words, if the difference in misses between two finite-control-state algorithms is unbounded, then it actually grows linearly in the length of the trace.

\newcommand{\mmax}{m_\textit{max}}
\newcommand{\lmin}{\len^*}

\begin{thm}\label{thm:lineargrowth}
	The leak ratio between two finite-control-state cache algorithms $\algorithm$ and $\algorithmprime$ grows linearly in the length of the trace if and only if the difference in misses between $\algorithm$ and $\algorithmprime$ is unbounded: 
\[\leakrat_{\algorithm,\algorithmprime}(\len), \leakrat_{\algorithmprime,\algorithm}(\len) \in \Omega(\len) \quad
  \Leftrightarrow 
\quad
\forall m \in \NN: \exists \trace \in \blockset^*: |\algorithm(\trace) - \algorithmprime(\trace)| > m.\]
\end{thm}
\begin{proof}
	Direction ``$\Rightarrow$'':
	Assume for a contradiction that there is an $\mmax \in \NN$, such that for all traces $\trace \in \blockset^*$: $|\algorithm(\trace)-\algorithmprime(\trace)| \leq \mmax$.
	As $\leakrat_{\algorithm,\algorithmprime}(\len) \in \Omega(\len)$ there must be an $\lmin$, such that $\leakrat_{\algorithm,\algorithmprime}(\lmin) > 2\cdot\mmax+1$.
	By the second part of Corollary~\ref{cor:singletrace}, there is a trace $\trace$ such that
	\[\mmax < \frac{\leakrat_{\algorithm,\algorithmprime}(\lmin)-1}{2} \leq |\algorithm(\trace)-\algorithmprime(\trace)|,\]
	which contradicts our assumption.

	Direction ``$\Leftarrow$'':
	We will prove that $\leakrat_{\algorithm,\algorithmprime}(\len) \in \Omega(\len)$. 
	The fact that $\leakrat_{\algorithmprime,\algorithm}(\len) \in \Omega(\len)$ follows by simply exchanging $\algorithm$ and $\algorithmprime$ because $|\algorithm(\trace) - \algorithmprime(\trace)| = |\algorithmprime(\trace) - \algorithm(\trace)|$.
	
	To prove that $\leakrat_{\algorithm,\algorithmprime}(\len) \in \Omega(\len)$, we have to show that there is a $k > 0$ and an $m_0 \in \NN$, such that $\forall m \in \NN,m > m_0: \leakrat_{\algorithm,\algorithmprime}(m) \geq k\cdot m$.
	
	By Theorem~\ref{thm:lineardiff}, we can conclude that there is an $f > 0$ and an $m_0' \in \NN$, such that $\forall m \in \NN, m > m_0': \exists \trace \in \blockset^m: |\algorithm(\trace) - \algorithmprime(\trace)| > f\cdot |\trace|$.
	Pick $k$ to be $f$ and $m_0$ to be $m_0'$.
	To prove the theorem it then remains to show that $\exists \trace \in \blockset^m : |\algorithm(\trace)-\algorithmprime(\trace)| > f \cdot |\trace|$ implies $\leakrat_{\algorithm,\algorithmprime}(m) > f \cdot m$.
	
	To this end, let $\trace \in \blockset^m$ be a trace that satisfies $|\algorithm(\trace)-\algorithmprime(\trace)| > f \cdot |\trace|$.
	Applying the first part of Corollary~\ref{cor:singletrace} yields $f \cdot |\trace| < |\algorithm(\trace)-\algorithmprime(\trace)| < \leakrat_{\algorithm,\algorithmprime}(m)$.
\end{proof}

\section{Related Work}\label{sec:related}
Leak competitiveness is inspired by work on competitive performance analysis. The notion of competitive analysis was first introduced in~\cite{sleator1985amortized}, where the authors bound the number of misses an online algorithm does on a trace of memory blocks in terms of the number of misses of an optimal offline algorithm~\cite{belady1966study}. In contrast to performance, there is no clear candidate for an optimal offline cache algorithm for security, because the best option would be not to cache memory blocks and be trivially non-interferent. This is why we base leak competitiveness on relative competitiveness~\cite{reineke2008relative}. Here, the number of misses one cache algorithm produces is compared with the number of misses of another algorithm, none of them necessarily optimal. 

The notion of leakage we use is based on concepts from quantitative information-flow analysis~\cite{ClarkHM07,m2012measuring,smith09}. They have been successfully used for detecting and quantifying side channels of program code~\cite{doychev2015cacheaudit,HeusserM10,koepfbasin07,newsome09}.

Concepts from quantitative information-flow analysis have been applied to the analysis of cache algorithms~\cite{canones2017}. Our work goes beyond this in two crucial aspects. First, we consider adversaries that measure overall execution time of a victim, whereas~\cite{canones2017} consider so-called {\em access-based adversaries} that gain information by probing the state of a shared cache after the victim's computation terminates. Second, our analysis is based on a comparison of cache algorithms on each program, whereas~\cite{canones2017} identifies the worst possible program for each.

We can, however, interpret some of the results of~\cite{canones2017} in terms of leak-competitiveness w.r.t. to an access-based adversary. The bound that governs leakage in this scenario is not the length of the trace but rather the number of memory blocks used by (i.e. the {\em footprint} of) the victim program. With this, one can read Propositions~6 and 7 of ~\cite{canones2017} as follows:
\begin{itemize}
\item for FIFO and LRU, the number of observations of an access-based adversary is bounded by a constant. This implies that the leak ratios of FIFO relative to LRU, and of LRU relative to FIFO, are in $\mathcal{O}(1)$;
\item for PLRU, the number of observations grows at least linearly with the footprint. This implies that the leak ratio of PLRU relative to FIFO and LRU, respectively, is in~$\Omega(n)$, whereas the leak ratio of FIFO and LRU relative to PLRU is in $\mathcal{O}(1)$.
\end{itemize}
Overall, these examples show that, unlike for time based adversaries, there are dominance relations for the security of cache algorithms with respect to access-based adversaries. We leave a detailed investigation of this case to future work.

Finally, a line of work focuses on secure cache architectures~\cite{he2017secure,zhang2014new}. They consider different architectures, either introducing some sort of partition on the cache or randomness in the replacement of memory blocks, and study their resilience against different kinds of cache side-channel attacks.

When it comes to timing attacks, they mention that, introducing some sort of randomness is the only way to reduce the vulnerability to leak information in this cases. This is because with deterministic cache algorithms, the attacker knows that the observation he obtains only depends on the victim's accesses to memory.
Our work acknowledges that this dependence is unavoidable for deterministic cache algorithms but tries to quantify how specific cache algorithms make the dependance less dangerous.

\section{Conclusions}\label{sec:conclusion}
We presented a novel approach to compare cache algorithms in terms of their vulnerability to side-channel attacks. Our core insight is that for leakage, as opposed to performance, there is no dominance relationship between any two cache algorithms, in the sense that one algorithm would outperform the other on all programs.

\subsection*{Acknowledgments}
We thank Pierre Ganty and the anonymous reviewers for their constructive feedback.

This work was supported by Microsoft Research through its PhD scholarship programme, a gift from Intel Corporation, Ram{\'o}n y Cajal grant RYC-2014-16766, Spanish projects TIN2015-70713-R DEDETIS and TIN2015-67522-C3-1-R TRACES, and Madrid regional project S2013/ICE-2731 N-GREENS.


\newcommand{\etalchar}[1]{$^{#1}$}

\appendix

\begin{scriptsize}
\begin{sidewaysfigure}[h!]
\vspace{127ex}
\centering
\begin{subfigure}{\textwidth}
$$
\begin{array}{rc}
\text{\scriptsize{LRU}}&\{ \invalid , \invalid \}\\
\text{\scriptsize{FIFO}}&\{ \invalid , \invalid \}
\end{array}
\overset{\mathrm{A}}{\underset{(1,1)}{\longrightarrow}}
\begin{array}{c}
\{ \mathrm{A} , \invalid \}\\
\{ \mathrm{A} , \invalid \}
\end{array}
\overset{\mathrm{B}}{\underset{(1,1)}{\longrightarrow}}
\begin{array}{c}
\{ \mathrm{B} , \textbf{A} \}\\
\{ \mathrm{B} , \textbf{A} \}
\end{array}
\overset{\mathrm{A}}{\underset{(0,0)}{\longrightarrow}}
\begin{array}{c}
\{ \mathrm{A} , \textbf{B} \}\\
\{ \mathrm{B} , \textbf{A} \}
\end{array}
\overset{\mathrm{C}}{\underset{(1,1)}{\longrightarrow}}
\begin{array}{c}
\{ \mathrm{C} , \textbf{A} \}\\
\{ \mathrm{C} , \textbf{B} \}
\end{array}
\overset{\mathrm{A}}{\underset{(0,1)}{\longrightarrow}}
\begin{array}{c}
\{ \mathrm{A} , \textbf{C} \}\\
\{ \mathrm{A} , \textbf{C} \}
\end{array}
\overset{\mathrm{C}}{\underset{(0,0)}{\longrightarrow}}
\begin{array}{c}
\{ \mathrm{C} , \textbf{A} \}\\
\{ \mathrm{A} , \textbf{C} \}
\end{array}
\overset{\mathrm{B}}{\underset{(1,1)}{\longrightarrow}}
\begin{array}{c}
\{ \mathrm{B} , \textbf{C} \}\\
\{ \mathrm{B} , \textbf{A} \}
\end{array}
\overset{\mathrm{B}}{\underset{(0,0)}{\longrightarrow}}
\begin{array}{c}
\{ \mathrm{B} , \textbf{C} \}\\
\{ \mathrm{B} , \textbf{A} \}
\end{array}
\overset{\mathrm{B}}{\underset{(0,0)}{\longrightarrow}}
\begin{array}{c}
\{ \mathrm{B} , \textbf{C} \}\\
\{ \mathrm{B} , \textbf{A} \}
\end{array}
\begin{array}{c}
\\
:\,(4,5)
\end{array}
$$

$$
\begin{array}{rc}
\text{\scriptsize{LRU}}&\{ \invalid , \invalid \}\\
\text{\scriptsize{FIFO}}&\{ \invalid , \invalid \}
\end{array}
\overset{\mathrm{A}}{\underset{(1,1)}{\longrightarrow}}
\begin{array}{c}
\{ \mathrm{A} , \invalid \}\\
\{ \mathrm{A} , \invalid \}
\end{array}
\overset{\mathrm{B}}{\underset{(1,1)}{\longrightarrow}}
\begin{array}{c}
\{ \mathrm{B} , \textbf{A} \}\\
\{ \mathrm{B} , \textbf{A} \}
\end{array}
\overset{\mathrm{A}}{\underset{(0,0)}{\longrightarrow}}
\begin{array}{c}
\{ \mathrm{A} , \textbf{B} \}\\
\{ \mathrm{B} , \textbf{A} \}
\end{array}
\overset{\mathrm{C}}{\underset{(1,1)}{\longrightarrow}}
\begin{array}{c}
\{ \mathrm{C} , \textbf{A} \}\\
\{ \mathrm{C} , \textbf{B} \}
\end{array}
\overset{\mathrm{D}}{\underset{(1,1)}{\longrightarrow}}
\begin{array}{c}
\{ \mathrm{D} , \textbf{C} \}\\
\{ \mathrm{D} , \textbf{C} \}
\end{array}
\overset{\mathrm{A}}{\underset{(1,1)}{\longrightarrow}}
\begin{array}{c}
\{ \mathrm{A} , \textbf{D} \}\\
\{ \mathrm{A} , \textbf{D} \}
\end{array}
\overset{\mathrm{A}}{\underset{(0,0)}{\longrightarrow}}
\begin{array}{c}
\{ \mathrm{A} , \textbf{D} \}\\
\{ \mathrm{A} , \textbf{D} \}
\end{array}
\overset{\mathrm{A}}{\underset{(0,0)}{\longrightarrow}}
\begin{array}{c}
\{ \mathrm{A} , \textbf{D} \}\\
\{ \mathrm{A} , \textbf{D} \}
\end{array}
\overset{\mathrm{A}}{\underset{(0,0)}{\longrightarrow}}
\begin{array}{c}
\{ \mathrm{A} , \textbf{D} \}\\
\{ \mathrm{A} , \textbf{D} \}
\end{array}
\begin{array}{c}
\\
:\, (5,5)
\end{array}
$$

$$
\begin{array}{rc}
\text{\scriptsize{LRU}}&\{ \invalid , \invalid \}\\
\text{\scriptsize{FIFO}}&\{ \invalid , \invalid \}
\end{array}
\overset{\mathrm{A}}{\underset{(1,1)}{\longrightarrow}}
\begin{array}{c}
\{ \mathrm{A} , \invalid \}\\
\{ \mathrm{A} , \invalid \}
\end{array}
\overset{\mathrm{B}}{\underset{(1,1)}{\longrightarrow}}
\begin{array}{c}
\{ \mathrm{B} , \textbf{A} \}\\
\{ \mathrm{B} , \textbf{A} \}
\end{array}
\overset{\mathrm{A}}{\underset{(0,0)}{\longrightarrow}}
\begin{array}{c}
\{ \mathrm{A} , \textbf{B} \}\\
\{ \mathrm{B} , \textbf{A} \}
\end{array}
\overset{\mathrm{C}}{\underset{(1,1)}{\longrightarrow}}
\begin{array}{c}
\{ \mathrm{C} , \textbf{A} \}\\
\{ \mathrm{C} , \textbf{B} \}
\end{array}
\overset{\mathrm{B}}{\underset{(1,0)}{\longrightarrow}}
\begin{array}{c}
\{ \mathrm{B} , \textbf{C} \}\\
\{ \mathrm{C} , \textbf{B} \}
\end{array}
\overset{\mathrm{A}}{\underset{(1,1)}{\longrightarrow}}
\begin{array}{c}
\{ \mathrm{A} , \textbf{B} \}\\
\{ \mathrm{A} , \textbf{C} \}
\end{array}
\overset{\mathrm{D}}{\underset{(1,1)}{\longrightarrow}}
\begin{array}{c}
\{ \mathrm{D} , \textbf{A} \}\\
\{ \mathrm{D} , \textbf{A} \}
\end{array}
\overset{\mathrm{D}}{\underset{(0,0)}{\longrightarrow}}
\begin{array}{c}
\{ \mathrm{D} , \textbf{A} \}\\
\{ \mathrm{D} , \textbf{A} \}
\end{array}
\overset{\mathrm{D}}{\underset{(0,0)}{\longrightarrow}}
\begin{array}{c}
\{ \mathrm{D} , \textbf{A} \}\\
\{ \mathrm{D} , \textbf{A} \}
\end{array}
\begin{array}{c}
\\
:\, (6,5)
\end{array}
$$

$$
\begin{array}{rc}
\text{\scriptsize{LRU}}&\{ \invalid , \invalid \}\\
\text{\scriptsize{FIFO}}&\{ \invalid , \invalid \}
\end{array}
\overset{\mathrm{A}}{\underset{(1,1)}{\longrightarrow}}
\begin{array}{c}
\{ \mathrm{A} , \invalid \}\\
\{ \mathrm{A} , \invalid \}
\end{array}
\overset{\mathrm{B}}{\underset{(1,1)}{\longrightarrow}}
\begin{array}{c}
\{ \mathrm{B} , \textbf{A} \}\\
\{ \mathrm{B} , \textbf{A} \}
\end{array}
\overset{\mathrm{A}}{\underset{(0,0)}{\longrightarrow}}
\begin{array}{c}
\{ \mathrm{A} , \textbf{B} \}\\
\{ \mathrm{B} , \textbf{A} \}
\end{array}
\overset{\mathrm{C}}{\underset{(1,1)}{\longrightarrow}}
\begin{array}{c}
\{ \mathrm{C} , \textbf{A} \}\\
\{ \mathrm{C} , \textbf{B} \}
\end{array}
\overset{\mathrm{B}}{\underset{(1,0)}{\longrightarrow}}
\begin{array}{c}
\{ \mathrm{B} , \textbf{C} \}\\
\{ \mathrm{C} , \textbf{B} \}
\end{array}
\overset{\mathrm{A}}{\underset{(1,1)}{\longrightarrow}}
\begin{array}{c}
\{ \mathrm{A} , \textbf{B} \}\\
\{ \mathrm{A} , \textbf{C} \}
\end{array}
\overset{\mathrm{C}}{\underset{(1,0)}{\longrightarrow}}
\begin{array}{c}
\{ \mathrm{C} , \textbf{A} \}\\
\{ \mathrm{A} , \textbf{C} \}
\end{array}
\overset{\mathrm{B}}{\underset{(1,1)}{\longrightarrow}}
\begin{array}{c}
\{ \mathrm{B} , \textbf{C} \}\\
\{ \mathrm{B} , \textbf{A} \}
\end{array}
\overset{\mathrm{B}}{\underset{(0,0)}{\longrightarrow}}
\begin{array}{c}
\{ \mathrm{B} , \textbf{C} \}\\
\{ \mathrm{B} , \textbf{A} \}
\end{array}
\begin{array}{c}
\\
:\, (7,5)
\end{array}
$$

$$
\begin{array}{rc}
\text{\scriptsize{LRU}}&\{ \invalid , \invalid \}\\
\text{\scriptsize{FIFO}}&\{ \invalid , \invalid \}
\end{array}
\overset{\mathrm{A}}{\underset{(1,1)}{\longrightarrow}}
\begin{array}{c}
\{ \mathrm{A} , \invalid \}\\
\{ \mathrm{A} , \invalid \}
\end{array}
\overset{\mathrm{B}}{\underset{(1,1)}{\longrightarrow}}
\begin{array}{c}
\{ \mathrm{B} , \textbf{A} \}\\
\{ \mathrm{B} , \textbf{A} \}
\end{array}
\overset{\mathrm{A}}{\underset{(0,0)}{\longrightarrow}}
\begin{array}{c}
\{ \mathrm{A} , \textbf{B} \}\\
\{ \mathrm{B} , \textbf{A} \}
\end{array}
\overset{\mathrm{C}}{\underset{(1,1)}{\longrightarrow}}
\begin{array}{c}
\{ \mathrm{C} , \textbf{A} \}\\
\{ \mathrm{C} , \textbf{B} \}
\end{array}
\overset{\mathrm{B}}{\underset{(1,0)}{\longrightarrow}}
\begin{array}{c}
\{ \mathrm{B} , \textbf{C} \}\\
\{ \mathrm{C} , \textbf{B} \}
\end{array}
\overset{\mathrm{A}}{\underset{(1,1)}{\longrightarrow}}
\begin{array}{c}
\{ \mathrm{A} , \textbf{B} \}\\
\{ \mathrm{A} , \textbf{C} \}
\end{array}
\overset{\mathrm{C}}{\underset{(1,0)}{\longrightarrow}}
\begin{array}{c}
\{ \mathrm{C} , \textbf{A} \}\\
\{ \mathrm{A} , \textbf{C} \}
\end{array}
\overset{\mathrm{B}}{\underset{(1,1)}{\longrightarrow}}
\begin{array}{c}
\{ \mathrm{B} , \textbf{C} \}\\
\{ \mathrm{B} , \textbf{A} \}
\end{array}
\overset{\mathrm{A}}{\underset{(1,0)}{\longrightarrow}}
\begin{array}{c}
\{ \mathrm{A} , \textbf{B} \}\\
\{ \mathrm{B} , \textbf{A} \}
\end{array}
\begin{array}{c}
\\
:\, (8,5)
\end{array}
$$
\caption{Example of a set of traces where LRU produces 5 observations and FIFO produces 1.}
\label{fig:example_a}
\end{subfigure}

\begin{subfigure}{\textwidth}
$$
\begin{array}{rc}
\text{\scriptsize{LRU}}&\{ \invalid , \invalid \}\\
\text{\scriptsize{FIFO}}&\{ \invalid , \invalid \}
\end{array}
\overset{\mathrm{A}}{\underset{(1,1)}{\longrightarrow}}
\begin{array}{c}
\{ \mathrm{A} , \invalid \}\\
\{ \mathrm{A} , \invalid \}
\end{array}
\overset{\mathrm{B}}{\underset{(1,1)}{\longrightarrow}}
\begin{array}{c}
\{ \mathrm{B} , \textbf{A} \}\\
\{ \mathrm{B} , \textbf{A} \}
\end{array}
\overset{\mathrm{A}}{\underset{(0,0)}{\longrightarrow}}
\begin{array}{c}
\{ \mathrm{A} , \textbf{B} \}\\
\{ \mathrm{B} , \textbf{A} \}
\end{array}
\overset{\mathrm{C}}{\underset{(1,1)}{\longrightarrow}}
\begin{array}{c}
\{ \mathrm{C} , \textbf{A} \}\\
\{ \mathrm{C} , \textbf{B} \}
\end{array}
\overset{\mathrm{B}}{\underset{(1,0)}{\longrightarrow}}
\begin{array}{c}
\{ \mathrm{B} , \textbf{C} \}\\
\{ \mathrm{C} , \textbf{B} \}
\end{array}
\overset{\mathrm{A}}{\underset{(1,1)}{\longrightarrow}}
\begin{array}{c}
\{ \mathrm{A} , \textbf{B} \}\\
\{ \mathrm{A} , \textbf{C} \}
\end{array}
\overset{\mathrm{A}}{\underset{(0,0)}{\longrightarrow}}
\begin{array}{c}
\{ \mathrm{A} , \textbf{B} \}\\
\{ \mathrm{A} , \textbf{C} \}
\end{array}
\overset{\mathrm{A}}{\underset{(0,0)}{\longrightarrow}}
\begin{array}{c}
\{ \mathrm{A} , \textbf{B} \}\\
\{ \mathrm{A} , \textbf{C} \}
\end{array}
\overset{\mathrm{A}}{\underset{(0,0)}{\longrightarrow}}
\begin{array}{c}
\{ \mathrm{A} , \textbf{B} \}\\
\{ \mathrm{A} , \textbf{C} \}
\end{array}
\begin{array}{c}
\\
:\, (5,4)
\end{array}
$$

$$
\begin{array}{rc}
\text{\scriptsize{LRU}}&\{ \invalid , \invalid \}\\
\text{\scriptsize{FIFO}}&\{ \invalid , \invalid \}
\end{array}
\overset{\mathrm{A}}{\underset{(1,1)}{\longrightarrow}}
\begin{array}{c}
\{ \mathrm{A} , \invalid \}\\
\{ \mathrm{A} , \invalid \}
\end{array}
\overset{\mathrm{B}}{\underset{(1,1)}{\longrightarrow}}
\begin{array}{c}
\{ \mathrm{B} , \textbf{A} \}\\
\{ \mathrm{B} , \textbf{A} \}
\end{array}
\overset{\mathrm{A}}{\underset{(0,0)}{\longrightarrow}}
\begin{array}{c}
\{ \mathrm{A} , \textbf{B} \}\\
\{ \mathrm{B} , \textbf{A} \}
\end{array}
\overset{\mathrm{C}}{\underset{(1,1)}{\longrightarrow}}
\begin{array}{c}
\{ \mathrm{C} , \textbf{A} \}\\
\{ \mathrm{C} , \textbf{B} \}
\end{array}
\overset{\mathrm{D}}{\underset{(1,1)}{\longrightarrow}}
\begin{array}{c}
\{ \mathrm{D} , \textbf{C} \}\\
\{ \mathrm{D} , \textbf{C} \}
\end{array}
\overset{\mathrm{A}}{\underset{(1,1)}{\longrightarrow}}
\begin{array}{c}
\{ \mathrm{A} , \textbf{D} \}\\
\{ \mathrm{A} , \textbf{D} \}
\end{array}
\overset{\mathrm{A}}{\underset{(0,0)}{\longrightarrow}}
\begin{array}{c}
\{ \mathrm{A} , \textbf{D} \}\\
\{ \mathrm{A} , \textbf{D} \}
\end{array}
\overset{\mathrm{A}}{\underset{(0,0)}{\longrightarrow}}
\begin{array}{c}
\{ \mathrm{A} , \textbf{D} \}\\
\{ \mathrm{A} , \textbf{D} \}
\end{array}
\overset{\mathrm{A}}{\underset{(0,0)}{\longrightarrow}}
\begin{array}{c}
\{ \mathrm{A} , \textbf{D} \}\\
\{ \mathrm{A} , \textbf{D} \}
\end{array}
\begin{array}{c}
\\
:\, (5,5)
\end{array}
$$

$$
\begin{array}{rc}
\text{\scriptsize{LRU}}&\{ \invalid , \invalid \}\\
\text{\scriptsize{FIFO}}&\{ \invalid , \invalid \}
\end{array}
\overset{\mathrm{A}}{\underset{(1,1)}{\longrightarrow}}
\begin{array}{c}
\{ \mathrm{A} , \invalid \}\\
\{ \mathrm{A} , \invalid \}
\end{array}
\overset{\mathrm{B}}{\underset{(1,1)}{\longrightarrow}}
\begin{array}{c}
\{ \mathrm{B} , \textbf{A} \}\\
\{ \mathrm{B} , \textbf{A} \}
\end{array}
\overset{\mathrm{A}}{\underset{(0,0)}{\longrightarrow}}
\begin{array}{c}
\{ \mathrm{A} , \textbf{B} \}\\
\{ \mathrm{B} , \textbf{A} \}
\end{array}
\overset{\mathrm{C}}{\underset{(1,1)}{\longrightarrow}}
\begin{array}{c}
\{ \mathrm{C} , \textbf{A} \}\\
\{ \mathrm{C} , \textbf{B} \}
\end{array}
\overset{\mathrm{A}}{\underset{(0,1)}{\longrightarrow}}
\begin{array}{c}
\{ \mathrm{A} , \textbf{C} \}\\
\{ \mathrm{A} , \textbf{C} \}
\end{array}
\overset{\mathrm{B}}{\underset{(1,1)}{\longrightarrow}}
\begin{array}{c}
\{ \mathrm{B} , \textbf{A} \}\\
\{ \mathrm{B} , \textbf{A} \}
\end{array}
\overset{\mathrm{C}}{\underset{(1,1)}{\longrightarrow}}
\begin{array}{c}
\{ \mathrm{C} , \textbf{B} \}\\
\{ \mathrm{C} , \textbf{B} \}
\end{array}
\overset{\mathrm{C}}{\underset{(0,0)}{\longrightarrow}}
\begin{array}{c}
\{ \mathrm{C} , \textbf{B} \}\\
\{ \mathrm{C} , \textbf{B} \}
\end{array}
\overset{\mathrm{C}}{\underset{(1,1)}{\longrightarrow}}
\begin{array}{c}
\{ \mathrm{C} , \textbf{B} \}\\
\{ \mathrm{C} , \textbf{B} \}
\end{array}
\begin{array}{c}
\\
:\, (5,6)
\end{array}
$$

$$
\begin{array}{rc}
\text{\scriptsize{LRU}}&\{ \invalid , \invalid \}\\
\text{\scriptsize{FIFO}}&\{ \invalid , \invalid \}
\end{array}
\overset{\mathrm{A}}{\underset{(1,1)}{\longrightarrow}}
\begin{array}{c}
\{ \mathrm{A} , \invalid \}\\
\{ \mathrm{A} , \invalid \}
\end{array}
\overset{\mathrm{B}}{\underset{(1,1)}{\longrightarrow}}
\begin{array}{c}
\{ \mathrm{B} , \textbf{A} \}\\
\{ \mathrm{B} , \textbf{A} \}
\end{array}
\overset{\mathrm{A}}{\underset{(0,0)}{\longrightarrow}}
\begin{array}{c}
\{ \mathrm{A} , \textbf{B} \}\\
\{ \mathrm{B} , \textbf{A} \}
\end{array}
\overset{\mathrm{C}}{\underset{(1,1)}{\longrightarrow}}
\begin{array}{c}
\{ \mathrm{C} , \textbf{A} \}\\
\{ \mathrm{C} , \textbf{B} \}
\end{array}
\overset{\mathrm{A}}{\underset{(0,1)}{\longrightarrow}}
\begin{array}{c}
\{ \mathrm{A} , \textbf{C} \}\\
\{ \mathrm{A} , \textbf{C} \}
\end{array}
\overset{\mathrm{C}}{\underset{(0,0)}{\longrightarrow}}
\begin{array}{c}
\{ \mathrm{C} , \textbf{A} \}\\
\{ \mathrm{A} , \textbf{C} \}
\end{array}
\overset{\mathrm{B}}{\underset{(1,1)}{\longrightarrow}}
\begin{array}{c}
\{ \mathrm{B} , \textbf{C} \}\\
\{ \mathrm{B} , \textbf{A} \}
\end{array}
\overset{\mathrm{C}}{\underset{(0,1)}{\longrightarrow}}
\begin{array}{c}
\{ \mathrm{C} , \textbf{B} \}\\
\{ \mathrm{C} , \textbf{B} \}
\end{array}
\overset{\mathrm{A}}{\underset{(1,1)}{\longrightarrow}}
\begin{array}{c}
\{ \mathrm{A} , \textbf{C} \}\\
\{ \mathrm{A} , \textbf{C} \}
\end{array}
\begin{array}{c}
\\
:\, (5,7)
\end{array}
$$
\caption{Example of a set of traces where FIFO produces 4 observations and LRU produces 1.}
\label{fig:example_b}
\end{subfigure}

\caption{Example of two sets of traces of memory blocks updating empty caches under cache algorithms LRU and FIFO.
Each configuration of the cache shows the blocks cached between brackets, the symbol $\invalid$ represents a line with no memory block cached in it and the bold face memory block is the next one to be evicted (the least recently used for LRU and the first to be cached for FIFO).
Each update of the cache configuration is labeled with the accessed memory block (over the arrow) and the number of misses it produces on LRU and FIFO respectively (under the arrow).
At the end of each trace we have the total number of misses the trace produces on LRU and FIFO respectively.}
\label{fig:example}
\end{sidewaysfigure}
\end{scriptsize}

\end{document}